\documentclass[lettersize,journal,10pt]{IEEEtran}
\usepackage{amsmath,amsfonts}
\usepackage{amsthm}
\newtheorem{lemma}{Lemma}
\newtheorem{remark}{Remark}
\newtheorem{proposition}{Proposition}
\usepackage{changes}
\usepackage{algorithm}
\usepackage{algpseudocode}
\usepackage{mathtools, cuted}
\usepackage{array}
\usepackage{amssymb}
\usepackage{textcomp}
\usepackage{stfloats}
\usepackage{url}
\usepackage{verbatim}
\usepackage{graphicx,subfigure}
\usepackage{cite}
\usepackage{tabularx}
\usepackage{hyperref}

\allowdisplaybreaks[4]
\makeatletter
\newcommand{\multiline}[1]{%
	\begin{tabularx}{\dimexpr\linewidth-\ALG@thistlm}[t]{@{}X@{}}
		#1
	\end{tabularx}
}
\makeatother
\newcommand\copyrighttext{%
	\footnotesize \textcopyright 2023 IEEE. Personal use of this material is permitted.
	Permission from IEEE must be obtained for all other uses, in any current or future
	media, including reprinting/republishing this material for advertising or promotional
	purposes, creating new collective works, for resale or redistribution to servers or
	lists, or reuse of any copyrighted component of this work in other works.
	DOI: \href{<https://ieeexplore.ieee.org/abstract/document/10373918>}{10.1109/LCOMM.2023.3344599}}
\newcommand\copyrightnotice{%
	\begin{tikzpicture}[remember picture,overlay]
		\node[anchor=south,yshift=10pt] at (current page.south) {\fbox{\parbox{\dimexpr\textwidth-\fboxsep-\fboxrule\relax}{\copyrighttext}}};
	\end{tikzpicture}%
}

\begin{document}

\title{Intelligent Reflecting Surfaces vs. Full-Duplex Relays: A Comparison in the Air}

\author{Qian Ding,~\IEEEmembership{}{Jie Yang,~\IEEEmembership{}}Yang Luo,~\IEEEmembership{}{and Chunbo Luo,~\IEEEmembership{Senior Member,~IEEE}}
	\thanks{Qian Ding, Jie Yang, Yang Luo and Chunbo Luo are with the School of Information and Communication Engineering, University of Electronic Science and Technology of China, Chengdu 611731, China (e-mail: dingqian@std.uestc.edu.cn; jieyang.std@gmail.com; luoyang@uestc.edu.cn; c.luo@ieee.org) (\textit{corresponding author: Chunbo Luo}). This research is partially funded by the S\&T Special Program of Huzhou (Grant No. 2022GZ14) and National Key Research and Development Program of China (2023YFC3806003, 2023YFC3806001).}
}

\markboth{Journal of \LaTeX\ Class Files,~Vol.~14, No.~8, August~2021}%
{Shell \MakeLowercase{\textit{et al.}}: A Sample Article Using IEEEtran.cls for IEEE Journals}


\maketitle
\copyrightnotice
\begin{abstract}
	This letter aims to provide a fundamental analytical comparison for the two major types of relaying methods: intelligent reflecting surfaces and full-duplex relays, particularly focusing on unmanned aerial vehicle communication scenarios. Both amplify-and-forward and decode-and-forward relaying schemes are included in the comparison. In addition, optimal 3D UAV deployment and minimum transmit power under the quality of service constraint are derived. Our numerical results show that IRSs of medium size exhibit comparable performance to AF relays, meanwhile outperforming DF relays under extremely large surface size and high data rates.
\end{abstract}


\section{Introduction}
\IEEEPARstart{I}{ntelligent} reflecting surface (IRS) is a revolutionizing architecture that can modify the amplitude and phase of incident signals via controllable reflection, to enhance the performance of wireless communication networks\cite{huang2019reconfigurable}. As a complementary device, IRS could be easily installed on the walls, ceilings and facades of buildings, benefiting from its low fabrication costs and light weight\cite{wu2019towards}. Although the ease of deployment is quite appealing, some fundamental limitations remain. First, deploying an IRS at a fixed place potentially reduces its coverage since it only serves a limited and fixed area. Second, line-of-sight (LoS) paths are typically difficult to maintain due to complex environments near ground. As a result, the signals will be severely attenuated after undergoing several reflections. Third, it is practically difficult to deploy IRSs considering environmental and civil constraints. To unlock the limitations of terrestrial IRSs, it is promising to deploy an IRS on an unmanned aerial vehicle (UAV), known as UAV-IRS, aerial IRS or UAV-mounted/borne/carried IRS\cite{you2021enabling}. Thanks to its mobility, UAV-IRS could bypass the obstacles to establish LoS links with ground users. More significantly, wider coverage range is realized through full-angle reflection towards the ground.

Despite the attractiveness of UAV-IRS, it is strongly encouraged to compare this integration with traditional full-duplex UAV-relay, given that previous literature mainly focused on the comparison in terrestrial scenarios\cite{huang2019reconfigurable,han2019large,bjornson2019intelligent,wu2019intelligent,gu2021performance,ye2021spatially,bazrafkan2023performance,goh2023comparative}. In \cite{di2020reconfigurable}, the authors reviewed the differences and similarities between full-duplex relays and IRSs. In \cite{huang2019reconfigurable} and \cite{han2019large}, the authors showed that IRSs could achieve 300\% higher energy efficiency against half-duplex amplify-and-forward (AF) relays while the authors of \cite{bjornson2019intelligent} pointed out that IRSs could outperform half-duplex decode-and-forward (DF) relays only under large surfaces and/or high data rates. Furthermore, AF full-duplex relays and IRSs were compared in MIMO systems by optimizing the beamforming matrices\cite{wu2019intelligent,gu2021performance}. Differently, \cite{ye2021spatially} showed that IRSs could outperform DF full-duplex relays in terms of outage probability and energy efficiency. More recently, \cite{bazrafkan2023performance} presented that IRSs had a huge performance loss compared to DF full-duplex relays. In\cite{goh2023comparative}, IRSs were compared with AF and DF relays, however, only some experimental results were provided, with no further theoretical analysis. In a nutshell, fair comparison between IRSs and relays including both AF and DF protocols, as well as mathematical analysis are needed and yet not provided.

Since IRSs could be interpreted as full-duplex\cite{wu2019towards}, we thus perceive full-duplex relays as fairer comparison targets. Our goal is to provide a fundamental analytical comparison between IRSs and full-duplex relays including AF and DF schemes, specifically in UAV communication networks. For fairness of comparison, we prove the optimal 3D deployment of UAV and derive the minimum transmit power under quality of service (QoS) constraint for IRS, AF and DF relays. The novel contribution will underpin the emerging paradigm of integrated perception, communication and control (IPCC) design. Further, numerical simulations are conducted to effectively compare IRSs against full-duplex relays and justify the proposed theorems.
\vspace{-0.5em}
\section{System Model}
\label{sec:system model}
To study the fundamental performance for fair comparison between the two types of relaying technologies, we consider a classic three-node system model comprised of a single-antenna source (S), a single-antenna destination (D) and an aerial node (R), i.e., UAV-IRS or UAV-relay, as shown in Fig. \ref{fig:three-node system}. The IRS is installed with $N=N_xN_y$ elements while the full-duplex relay is equipped with two uniform planar arrays (UPAs) of size $N$ separately used for transmission and reception. It is assumed that the direct link between S and D suffers from deep fading and thus is negligible\cite{bjornson2019intelligent,bazrafkan2023performance}. The UAV helps establish LoS transmission links with S and D in order to relay signals. The coordinates of S, D and R are given by $(0,0,0)$, $(L,0,0)$ and $(x_U,y_U,h_U)$, respectively. Hence the deterministic air-to-ground channel gain from S to R and that from R to D are, respectively, expressed by\cite{zeng2016throughput,Ding2023maximization}
\vspace{-0.2em}
\begin{equation}
	g_{SR} = \beta_0d_{SR}^{-2},\quad g_{RD}=\beta_0d_{RD}^{-2},\label{eq:channel gain}
\end{equation}
with $\beta_0$ denoting the channel gain at a reference distance of 1m, $d_{SR}=\sqrt{x_U^2 + y_U^2 +h_U^2}$ and $d_{RD}=\sqrt{(x_U-L)^2 + y_U^2 +h_U^2}$. Further, the channel vectors from S to R and that from R to D are given by, respectively,
\begin{align}
	\mathbf{h}_{SR}=\sqrt{g_{SR}}\mathbf{a}_{SR},\quad \mathbf{h}_{RD}=\sqrt{g_{RD}}\mathbf{a}_{RD},
\end{align}
with $\mathbf{a}_{i}=[1,\ldots,e^{-j\frac{2\pi}{\lambda}d\sin\theta_{i}[(n_x-1)\cos\phi_{i}+(n_y-1)\sin\phi_{i}]},\ldots,$ $e^{-j\frac{2\pi}{\lambda}d\sin\theta_{i}[(N_x-1)\cos\phi_{i}+(N_y-1)\sin\phi_{i}]}]^T, i\in\{SR,RD\}$. $\lambda$ represents the wavelength and $d$ denotes antenna separation. $n_x$ and $n_y$ denote the indexes of antenna along horizontal and vertical directions, respectively. $\theta_{i}$ and $\phi_{i}$ denote the corresponding elevation and azimuth angles, respectively.
\begin{figure}[t]
	\centering
	\includegraphics[width=0.8\columnwidth]{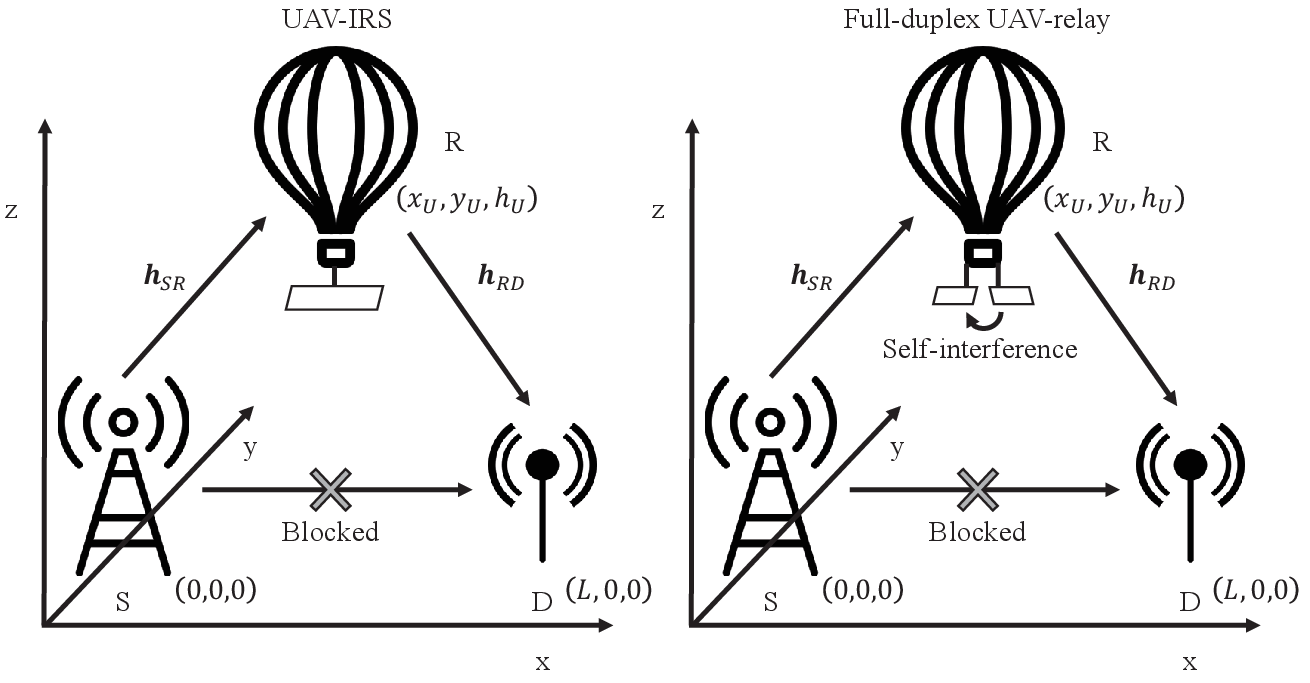}
	\caption{Illustration of three-node systems}
	\label{fig:three-node system}
	\vspace{-1em}
\end{figure}
\vspace{-0.5em}
\subsection{IRS-assisted networks}
In this setup, we let diagonal matrix $\mathbf{\Psi}=\text{diag}(e^{j\psi_1},\ldots,e^{j\psi_N})$ represent the passive beamforming of IRS, and thus the received signal at D is denoted by $y_D^{(IRS)}=\mathbf{h}_{RD}^H\mathbf{\Psi}\mathbf{h}_{SR}\sqrt{p_S^{(IRS)}}s + n_D$, where $p_S^{(IRS)}$ denotes the transmit power of S, $s$ represents the transmit signal with unit power and $n_D\sim \mathcal{CN}(0,\sigma^2)$ is the noise at D. Thereby, the corresponding capacity is
\begin{equation}
	\begin{aligned}
		R_{IRS} &= \log_2\bigg(1+\frac{p_S^{(IRS)}|\mathbf{h}_{RD}^H\mathbf{\Psi}\mathbf{h}_{SR}|^2}{\sigma^2}\bigg)\\
	&\overset{(a)}{=}\log_2\bigg(1+\frac{N^2p_S^{(IRS)}g_{SR}g_{RD}}{\sigma^2}\bigg),\label{eq:rate of IRS}
	\end{aligned}
\end{equation}
where $(a)$ holds by adopting coherent phase shifting scheme $\psi_n=j\frac{2\pi}{\lambda}d\{\sin\theta_{SR}[(n_x-1)\cos\phi_{SR}+(n_y-1)\sin\phi_{SR}]-\sin\theta_{RD}[(n_x-1)\cos\phi_{RD}+(n_y-1)\sin\phi_{RD}]\},\forall n=n_xn_y$\cite{bjornson2019intelligent}. The beam alignment for the high-mobility UAV-IRS can be achieved using the framework in \cite{chen2022reconfigurable}. As expected, the received power increases in the order of $\mathcal{O}(N^2)$\cite{wu2019towards}.
\vspace{-0.5em}
\subsection{AF full-duplex relay-assisted networks}
In this case, the full-duplex relay simultaneously receives and transmits signals with AF scheme. Hence, the received signal at relay is given by $\mathbf{y}_R^{(AF)}=\sqrt{p_{S}^{(AF)}}\mathbf{h}_{SR}s+\sqrt{\rho p_R^{(AF)}}\mathbf{\tilde{y}}_R^{(AF)}+\mathbf{n}_R$,
where $p_{S}^{(AF)}$ and $p_{R}^{(AF)}$ denote the transmit power of S and AF relay, respectively. $\rho$ characterizes the effects of self-interference cancellation, $\mathbf{n}_R\sim \mathcal{CN}(0,\sigma^2\mathbf{I}_N)$ is the noise, and $\mathbf{\tilde{y}}_R^{(AF)}$ is a delayed and distorted replica of $\mathbf{y}_R^{(AF)}$. Without loss of generality, we assume $\mathbf{\tilde{y}}_R^{(AF)}\sim \mathcal{CN}(0,\mathbf{I}_N)$\cite{sabharwal2014band}. Then the received signal is processed with combiner $\mathbf{v}$ and forward with precoder $\mathbf{w}$, where $\|\mathbf{v}\| = 1$ and $\|\mathbf{w}\| = 1$.  Therefore, the received signal at D is expressed by
\vspace{-0.5em}
\begin{equation}
	\begin{aligned}
		y_D^{(AF)}&\!=\!\sqrt{p_R^{(AF)}}\mathbf{h}_{RD}^H\mathbf{w}\mathbf{v}^H\mathbf{y}_R^{(AF)} + n_D\\
	&\!=\!\sqrt{\!p_R^{(AF)}p_{S}^{(AF)}}\mathbf{h}_{RD}^H\mathbf{w}\mathbf{v}^H\mathbf{h}_{SR}s\\
	&\!+\!\sqrt{\!\rho}p_R^{(AF)}\mathbf{h}_{RD}^H\mathbf{w}\mathbf{v}^H\tilde{\mathbf{y}}_R^{(AF)}\!\!+\!\! \sqrt{\!p_R^{(AF)}}\mathbf{h}_{RD}^H\mathbf{w}\mathbf{v}^H\mathbf{n}_R \!+\! n_D.
	\end{aligned}
\end{equation}
\begin{figure*}[t]
	\begin{equation}
		\begin{aligned}
			R_{AF}&=\!\mathbb{E}\left[\log_2\left(1+\frac{p_R^{(AF)}p_{S}^{(AF)}|\mathbf{h}_{RD}^H\mathbf{w}\mathbf{v}^H\mathbf{h}_{SR}|^2}{\rho (p_R^{(AF)})^2|\mathbf{h}_{RD}^H\mathbf{w}\mathbf{v}^H\tilde{\mathbf{y}}_R^{(AF)}|^2+p_R^{(AF)}|\mathbf{h}_{RD}^H\mathbf{w}\mathbf{v}^H\mathbf{n}_R|^2+|n_D|^2}\right)\right]\\
			&\overset{(b)}{\leq}\!\log_2\!\left(\!1\!+\!\frac{p_R^{(AF)}p_{S}^{(AF)}\mathbb{E}\left[|\mathbf{h}_{RD}^H\mathbf{w}\mathbf{v}^H\mathbf{h}_{SR}|^2\right]}{(N\rho (p_R^{(AF)})^2 + N\sigma^2p_R^{(AF)})\mathbb{E}[|\mathbf{h}_{RD}^H\mathbf{w}|^2]+\sigma^2}\!\right)\overset{(c)}{=}\log_2\!\left(\!1\!+\!\frac{N^2p_R^{(AF)}p_{S}^{(AF)}g_{SR}g_{RD}}{(N\rho (p_R^{(AF)})^2 + N\sigma^2p_R^{(AF)})g_{RD}+\sigma^2}\!\right).\label{eq:rate for AF relay}
		\end{aligned}
	\end{equation}
	\hrule
	\vspace{-1em}
\end{figure*}
The achievable rate of AF relay-aided networks is derived as \eqref{eq:rate for AF relay} at the top of next page, where $(b)$ holds due to Jensen's inequality and the fact that $\mathbb{E}[\|\tilde{\mathbf{y}}_R^{(AF)}\|^2]=N,\mathbb{E}[\|\mathbf{n}_R\|]=\sigma^2N, \mathbb{E}[|n_R|^2]=\sigma^2,\|\mathbf{w}\|^2 = 1$ and $\|\mathbf{v}\|^2=1$; $(c)$ holds by using maximum ratio transmission and combing scheme $\mathbf{v}=\frac{\mathbf{h}_{SR}}{\|\mathbf{h}_{SR}\|}$, $\mathbf{w}=\frac{\mathbf{h}_{RD}}{\|\mathbf{h}_{RD}\|}$\cite{wu2019intelligent} and the fact that $\mathbb{E}[\|\mathbf{h}_{SR}\|^2]=Ng_{SR}$ and $\mathbb{E}[\|\mathbf{h}_{RD}\|^2]=Ng_{RD}$.
\vspace{-0.5em}
\subsection{DF full-duplex relay-assisted networks}
With DF relaying protocol, the relay first decodes the received signal, which is given by $\mathbf{y}_R^{(DF)}=\sqrt{p_S^{(DF)}}\mathbf{h}_{SR}s+\sqrt{\rho p_R^{(DF)}}\tilde{\mathbf{s}}+\mathbf{n}_R$,
where $p_S^{(DF)}$ and $p_R^{(DF)}$ denote the transmit power of S and DF relay, respectively. $\tilde{\mathbf{s}}\sim \mathcal{CN}(0,\mathbf{I}_N)$ is a delayed and distorted replica of $\mathbf{w}s$. After decoding the information, DF relay encodes it for further transmission and the signal received by D is $y_D^{(DF)}=\sqrt{p_R^{(DF)}}\mathbf{h}^H_{RD}\mathbf{w}s+n_D$.
Following the similar analysis to \eqref{eq:rate for AF relay}, the achievable rate of DF relay-aided networks is
\begin{equation}\color{blue}
	R_{DF}\leq\log_2\bigg(1+\min\bigg(\frac{p_S^{(DF)}g_{SR}}{\rho p_R^{(DF)}+\sigma^2},\frac{Np_R^{(DF)}g_{RD}}{\sigma^2}\bigg)\bigg).\label{eq:rate for DF relay}
\end{equation}
\vspace{-2em}
\section{3D UAV deployment optimization}
\label{sec:optimal UAV deployment}
In this section, we provide rigorous proof for optimal UAV deployment\footnote{The deployment strategy in this letter is also applicable to terrestrial scenarios if LoS paths are not blocked and much stronger than non-LoS (NLoS) paths.}. It is assumed that UAV could move horizontally in a wide area $\mathcal{A}$ covering S and D while it could fly vertically at the altitude ranging from $h_{min}$ to $h_{max}$, where $h_{min}$ denotes the minimum height to maintain an LoS path, and $h_{max}$ means the maximum tolerable height.
\vspace{-0.5em}
\subsection{IRS-assisted networks}
Here we start with IRS to draw some useful insights. Considering the deployment constraints and substituting \eqref{eq:channel gain} into \eqref{eq:rate of IRS}, we can express the UAV positioning problem as
\begin{subequations}
	\begin{align}
		\max_{x_U,y_U,h_U} &R_{IRS}=\log_2\bigg(1+\frac{N^2p_S^{(IRS)}\beta_0^2}{\sigma^2d_{SR}^{2}d_{RD}^{2}}\bigg)\label{obj:IRS deployment}\\
		\text{s.t.} \quad &(x_U,y_U)\in \mathcal{A},h_{min}\leq h_U \leq h_{max}.\label{cons:UAV deployment for IRS}
	\end{align}\label{pro:UAV deployment for IRS}
\end{subequations}
We observe that only the denominator in \eqref{obj:IRS deployment} is related to the UAV position. Thus problem \eqref{pro:UAV deployment for IRS} is reduced to
\vspace{-0.5em}
\begin{subequations}
	\begin{align}
		\min_{x_U,y_U,h_U} &f(x_U,y_U,h_U)=d_{SR}^{2}d_{RD}^{2}\nonumber\\
		&=(x_U^2 + y_U^2 +h_U^2)\big((x_U-L)^2 + y_U^2 +h_U^2\big)\label{obj:UAV deployment for IRS}\\
		\text{s.t.} \quad &\eqref{cons:UAV deployment for IRS}.
	\end{align}\label{pro:UAV deployment for IRS:simplified}
\end{subequations}
\vspace{-1em}
\begin{lemma}
The optimal solution to problem \eqref{pro:UAV deployment for IRS:simplified} is
\begin{align}
	&x_U^*=\left\{\begin{aligned}
		&\frac{L}{2}, \quad\text{if}\ \ L\leq 2h_{min}\\
		&\frac{L}{2} \pm\sqrt{\frac{L^2}{4}-h_{min}^2},\quad\text{if}\ \ L> 2h_{min}
	\end{aligned}
	\right.,\\
	&y_U^*=0,h_U^*=h_{min}.
\end{align}\label{lemma:optimal UAV deployment for IRS}
\end{lemma}
\vspace{-3em}
\begin{proof}
	See Appendix \ref{proof:lemma 1}.
\end{proof}
\vspace{-1em}
%
%
%
\begin{remark}
	Lemma \ref{lemma:optimal UAV deployment for IRS} states that UAV should hover at the midpoint of S and D if its flight altitude is relatively high. Otherwise, the UAV should hover near to either S or D when its flight altitude is relatively low.
	\label{remark:remark for IRS}
\end{remark}
\vspace{-2em}
\subsection{AF full-duplex relay-assisted networks}
For AF relays, we substitute \eqref{eq:channel gain} into \eqref{eq:rate for AF relay} and follow the steps in the IRS case. A simplified UAV deployment problem is expressed by
\begin{align}
	\min_{x_U}\quad \xi_1 \tilde{g}(x_U) + \xi_2 \tilde{f}(x_U), \quad \text{s.t.} \quad x_U\in[0,L],\label{pro:UAV deployment for AF relay:simplified}
\end{align}
where $\xi_1=\frac{\rho p_R^{(AF)}+\sigma^2}{\added{N}\beta_0p_S^{(AF)}}$, $\xi_2=\frac{\sigma^2}{N^2\beta_0^2p_S^{(AF)}p_{R}^{(AF)}}$, $\tilde{g}(x_U)=x_U^2 + h_{min}^2$ and $\tilde{f}(x_U)$ is defined in Lemma \ref{lemma:optimal UAV deployment for IRS}.
\begin{lemma}
	The optimal solution to problem \eqref{pro:UAV deployment for AF relay:simplified} satisfies $\tilde{x}^*_U \in [0, x_U^*]$, with $x_U^*$ being defined in Lemma \ref{lemma:optimal UAV deployment for IRS}.
\label{lemma:optimal UAV deployment of AF relay}
\end{lemma}
\begin{proof}
	First, the solution to minimize $\xi_2 \tilde{f}(x_U)$ is defined in Lemma \ref{lemma:optimal UAV deployment for IRS}. Further, we could see that $\tilde{g}(x_U)$ is increasing over $x_U\in[0,L]$, and as a result the optimal solution of $\xi_1 \tilde{g}(x_U) + \xi_2 \tilde{f}(x_U)$ is no greater than that of sole $\xi_2 \tilde{f}(x_U)$, which proves Lemma \ref{lemma:optimal UAV deployment of AF relay}.
\end{proof}
\vspace{-1em}
Although we have Lemma \ref{lemma:optimal UAV deployment of AF relay}, it is still non-trivial to write the expression of optimal solution. Fortunately, it could be readily proved that $\tilde{f}(x_U)$ and $\tilde{g}(x_U)$ are convex over the intervals of interest, and thus $\xi_1 \tilde{g}(x_U) + \xi_2 \tilde{f}(x_U)$ is also convex. Therefore, one-dimensional searching methods such as golden section search can be adopted to obtain the optimal solution efficiently.
\begin{remark}
	Lemma \ref{lemma:optimal UAV deployment of AF relay} reveals that the UAV should never be placed at D side in the AF relay case. Moreover, if we define $\xi \triangleq \frac{\xi_1}{\xi_2}=\frac{\added{N}\beta_0p_R^{(AF)}(\rho p_R^{(AF)} + \sigma^2)}{\sigma^2}$, we could intuitively deduce that the UAV deployment only depends on $p_R^{(AF)}$, $N$, $L$ and $h_{min}$. If $p_R^{(AF)}\rightarrow +\infty$, then $\xi\gg 1$ and $\tilde{g}(x_U)$ becomes dominant. To minimize $\tilde{g}(x_U)$, we have $\tilde{x}_U^*\rightarrow 0$. On the other hand, if $p_R^{(AF)}\rightarrow 0$, then $\xi\ll 1$ and $\tilde{f}(x_U)$ becomes dominant. In this case, $\tilde{x}_U^*$ will be close to Lemma \ref{lemma:optimal UAV deployment for IRS}.
	\label{remark:remark for AF relay}
\end{remark}
\vspace{-0.5em}
\subsection{DF full-duplex relay-assisted networks}
For ease of notation, we denote \added{$\mu_1 = \frac{p_S^{(DF)}}{\rho p_R^{(DF)} + \sigma^2}$} and $\mu_2=\frac{\added{N}p_R^{(DF)}}{\sigma^2}$. Our goal is to maximize the minimum value between $\mu_1g_{SR}$ and $\mu_1g_{RD}$. Note that in the DF case, we similarly have $y_U^*=0$ and $h_U^*=h_{min}$, which leads to a much simpler problem as follows,
\begin{align}
	\max_{x_U} \min\left\{\frac{\mu_1}{\tilde{g}(x_U)},\frac{\tilde{g}(x_U)\mu_2}{\tilde{f}(x_U)}\right\}\quad \text{s.t.}\ x_U\in [0,L],\label{pro:UAV deployment for DF relay}
\end{align}
with $\tilde{g}(x_U)$ and $\tilde{f}(x_U)$ being defined in problem \eqref{pro:UAV deployment for AF relay:simplified}.
\begin{lemma}
	The optimal solution to problem \eqref{pro:UAV deployment for DF relay} is
	\begin{align}
		&x_U^*=\left\{\begin{aligned}
			&0, \quad\text{if}\ \ \mu_1\leq \nu\mu_2\\
			&L, \quad\text{if}\ \ \mu_1\geq\frac{1}{\nu}\mu_2\\
			&\frac{L}{2}, \quad\text{if}\ \ \mu_1 = \mu_2\\
			&kL\!+\!\sqrt{k^2L^2\!-kL^2\!-h_{min}^2},\ \text{if}\ \ \nu\mu_2<\mu_1 < \mu_2\\
			&kL\!-\!\sqrt{k^2L^2\!-kL^2\!-h_{min}^2},\ \text{if}\ \ \mu_2<\mu_1 < \frac{1}{\nu}\mu_2
		\end{aligned}
		\right.,
	\end{align}\label{lemma:optimal UAV deployment for DF relay}
with $\nu=\frac{h_{min}^2}{L^2+h_{min}^2}$ and $k=\frac{\mu_1}{\mu_1 - \mu_2}$.
\end{lemma}
\vspace{-1em}
\begin{proof}
	First, it can be seen that $\frac{\mu_1}{\tilde{g}(x_U)}$ is always no greater than $\frac{\tilde{g}(x_U)\mu_2}{\tilde{f}(x_U)}$ if $\mu_1\leq \nu\mu_2$. To maximize $\frac{\mu_1}{\tilde{g}(x_U)}$, we shall have $x_U^* = 0$. Similarly, we need to maximize $\frac{\tilde{g}(x_U)\mu_2}{\tilde{f}(x_U)}$ when $\mu_1\geq\frac{1}{\nu}\mu_2$, which yields $x_U^*=L$. When $\nu\mu_2<\mu_1<\frac{1}{\nu}\mu_2$, the optimal solution is obtained by solving equation $\frac{\mu_1}{\tilde{g}(x_U)}=\frac{\tilde{g}(x_U)\mu_2}{\tilde{f}(x_U)}$. If $\mu_1 = \mu_2$, it is a linear equation with solution $x_U^*=\frac{L}{2}$. Otherwise, it is a quadratic equation with its solutions defined in Lemma \ref{lemma:optimal UAV deployment for DF relay}.
\end{proof}
\begin{remark}
	Lemma \ref{lemma:optimal UAV deployment for DF relay} indicates that UAV deployment in the DF case depends on $p_S^{(DF)}$, $p_R^{(DF)}$, $N$, $L$ and $h_{min}$. If $\mu_1 >\mu_2$, the UAV needs to be deployed at D side. Differently, under conditions $\mu_1 <\mu_2$, the UAV ought to be placed near S.
	\label{remark:remark for DF relay}
\end{remark}
\section{Transmit power minimization under QoS constraint}
\label{sec:minimum transmit power with rate constraints}

Considering low energy consumption, the transmit power ought to be minimized, without violating the QoS requirement, say $R_i\geq R_0,i\in\{IRS,AF,DF\}$. From \eqref{eq:rate of IRS}, \eqref{eq:rate for AF relay} and \eqref{eq:rate for DF relay}, it can be seen that data rates are non-decreasing with respect to transmit power. Therefore, the equality must hold at the optimal point and our goal transforms into finding the minimum transmit power under $R_i=R_0,i\in\{IRS,AF,DF\}$, which is concluded in Proposition \ref{prop:minimum transmit power}.
\begin{proposition}
	\label{prop:minimum transmit power}
	To achieve a data rate $R_0$, the minimum transmit power for IRS-aided networks scales down with $N^2$, denoted by
	\begin{equation}		
		p_S^{(IRS)}=\frac{\sigma^2(2^{R_0}-1)}{N^2g_{SR}g_{RD}}.\label{eq:minimum transmit power for IRS}
	\end{equation}
	The AF relay case requires the transmit power
	\vspace{-0.5em}
\begin{equation}
			p_S^{(\!AF\!)}\!\!+\!p_R^{(\!AF\!)}\!=\!2\sqrt{\!\frac{\sigma^2\!(\rho(2^{\!R_0}\!-1)\!+\!g_{SR})\!(2^{\!R_0}\!-\!1)}{N^2g_{SR}^2g_{RD}}}\!+\!\frac{\sigma^2(2^{\!R_0}\!-\!1)}{g_{SR}}.\label{eq:minimumtransmit power for AF relay}
	\end{equation}
The DF relay case requires the transmit power
\begin{equation}\color{blue}
	p_S^{(DF)}\!+\! p_R^{(DF)}\!=\!\frac{\sigma^2\rho(2^{R_0}\!-\!1)^2}{Ng_{SR}g_{RD}}\!+\!\frac{\sigma^2(2^{R_0}\!-\!1)}{g_{SR}}\!+\!\frac{\sigma^2(2^{R_0}\!-\!1)}{Ng_{RD}}.\label{eq:minimumtransmit power for DF relay}
\end{equation}
\vspace{-1em}
\end{proposition}
\begin{proof}
	See Appendix \ref{proof:proposition 1}. 
\end{proof}
\begin{figure}[t]
	\centering
	\includegraphics[width=.6\columnwidth]{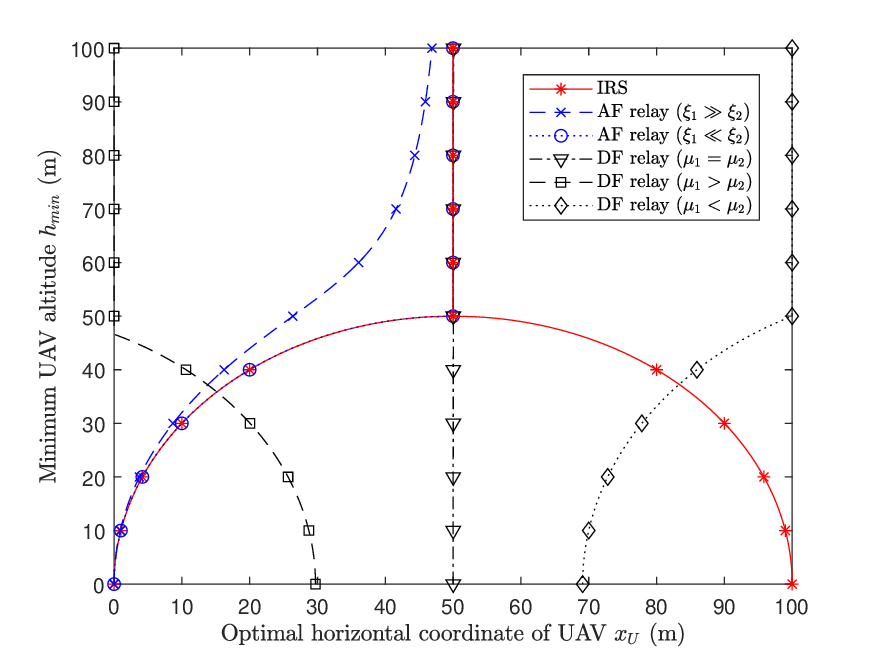}\\
	\caption{Optimal UAV deployment for IRS, AF and DF relays under $N=10$ and $L=100$m}
	\label{fig:optimal UAV deployment}
	\vspace{-1em}
\end{figure}
\vspace{-1em}
\begin{figure}[t]
	\centering
	\includegraphics[width=.6\columnwidth]{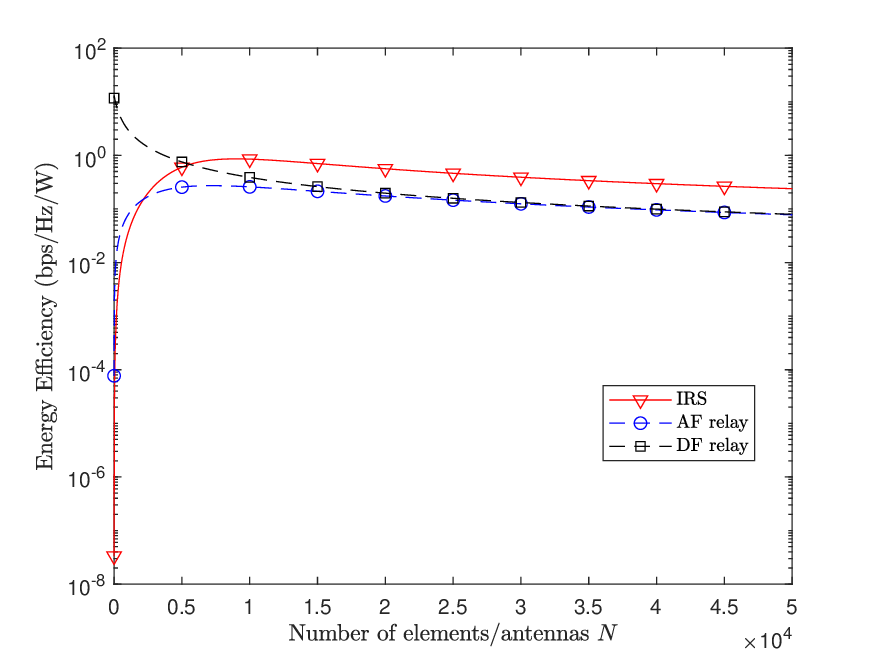}
	\caption{Energy efficiency versus $N$ under $L=200$m and $R_0=4$bps/Hz/W}
	\label{fig:energy efficiency versus N}
	\vspace{-1em}
\end{figure}
\section{Numerical Results}
\label{sec:simulation results}
In this section, we present numerical results to gain useful insights. We adopt an alternating optimization method\cite{huang2019reconfigurable} to handle the coupled nature of UAV deployment and transmit power. This approach involves iteratively updating each variable while keeping the other fixed. The iteration process continues until the decrease in transmit power falls below a predefined threshold $\epsilon$. The default parameter settings are $\beta_0=-60$dB\cite{zeng2016throughput}, $\sigma^2=-110$dBm\cite{zeng2016throughput}, $\rho = -110$dB\cite{sabharwal2014band} and $\epsilon=10^{-3}$\cite{Ding2023maximization}.

First, we show the optimal $x_U$ versus $h_{min}$ in Fig. \ref{fig:optimal UAV deployment}. The transmit power at S is set to 20dBm\cite{Ding2023maximization}. As expected in Remark \ref{remark:remark for IRS}, when $h_{min}$ is small, IRSs could be deployed near S or D, otherwise they should be placed at the midpoint. Differently, when $\xi_1\ll \xi_2$ ($p_R^{(AF)}=10$dBm), the placement of AF relay is similar to that of IRS at S side. When $\xi_1\gg \xi_2$ ($p_R^{(AF)}=50$dBm), the UAV ought to be deployed closer to S (see Remark \ref{remark:remark for AF relay}). Further, in the DF relay case, we can see that the UAV could be deployed at midpoint ($p_R^{(DF)}=4.3$dBm), S side ($p_R^{(DF)}=8.5$dBm) or D side ($p_R^{(DF)}=0$dBm), depending on the value of $\mu_1$ and $\mu_2$ (see Remark \ref{remark:remark for DF relay}).

In the sequel, we analyze the energy efficiency for each system. The total power consumption contains transmit power and hardware dissipation power. We denote $\omega$ as the drain efficiency. In the IRS-aided case, energy model is expressed by $P^{(IRS)}=\omega^{-1}p_S^{(IRS)} + P_S + P_D + NP_e$,	where $P_S$, $P_D$ and $P_e$ denote the power consumption of S, D and each element of IRS, respectively. For AF and DF relays, we can write their system power as $P^{(i)} = \omega^{-1}(p_S^{(i)}+p_R^{(i)}) + P_S + P_D + P^{(i)}_R + 2NP_a$, with $P^{(i)}_R,i\in\{AF,DF\}$ denoting the hardware dissipation power of relay and $P_a$ representing the power consumption of each antenna. Specifically, $P_S$, $P_D$ and $P_R^{(i)},i\in\{AF,DF\}$ are set equally to 100mW\cite{bjornson2019intelligent}, $P_e$ is 0.33mW\cite{bazrafkan2023performance}, $P_a$ is 0.5mW\cite{bazrafkan2023performance} and drain efficiency $\omega$ is 0.5\cite{bjornson2019intelligent}. Energy efficiency is computed by the ratio of data rates to total power consumption. It is shown in Fig. \ref{fig:energy efficiency versus N} that energy efficiency of both IRS and AF relay first increases and then decreases while that of DF relay keeps declining. Moreover, as $N$ increases, the energy efficiency of DF relay will get closer to AF relay's since $NP_a$ becomes dominant.
\begin{figure*}[t]
	\centering
	\subfigure[Minimum transmit power versus distance $L$ under $h_{min}=100$m and $R_0=4$bps/Hz]{
		\includegraphics[width=.6\columnwidth]{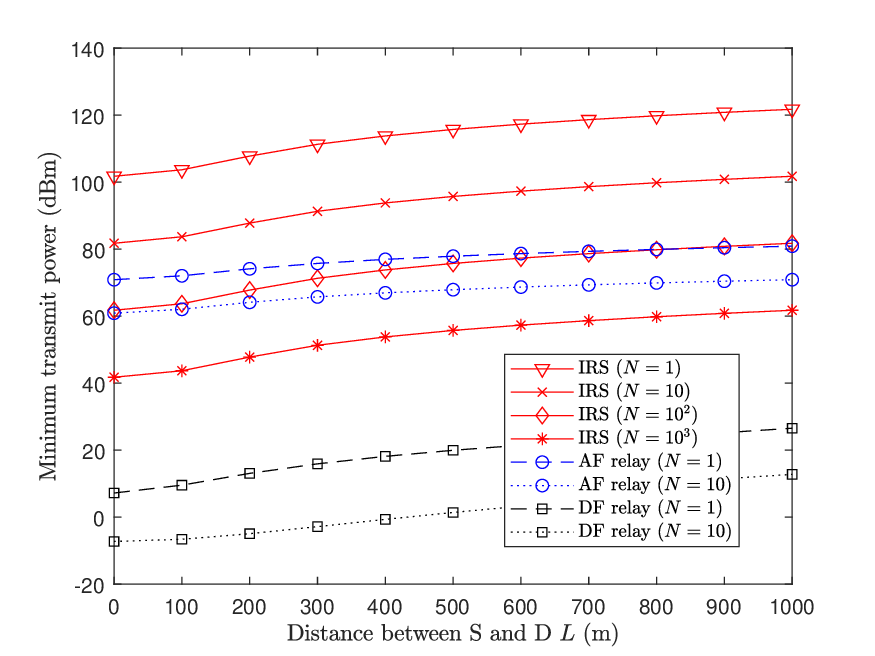}
		\label{fig:minimum power versus distance}}
	\hfil
	\centering
	\subfigure[Energy efficiency versus achievable rate $R_0$ under $L=200$m and $h_{min}=100$m]{
		\includegraphics[width=.6\columnwidth]{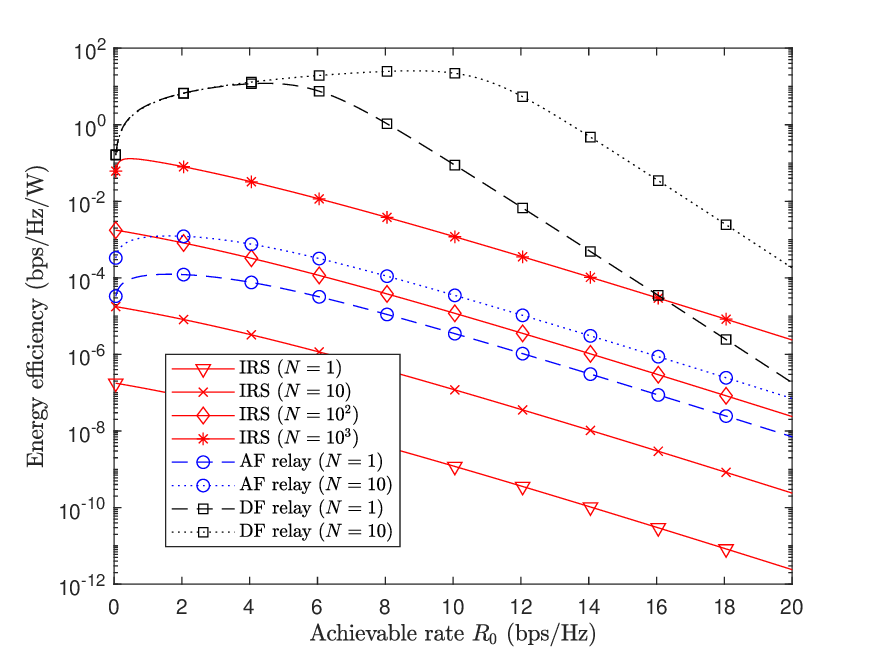}
		\label{fig:energy efficiency versus achievable rate}}
	\hfil
	\centering
	\subfigure[Energy efficiency versus $h_{min}$ under $L=200$m and $R_0=4$bps/Hz/W]{
		\includegraphics[width=.6\columnwidth]{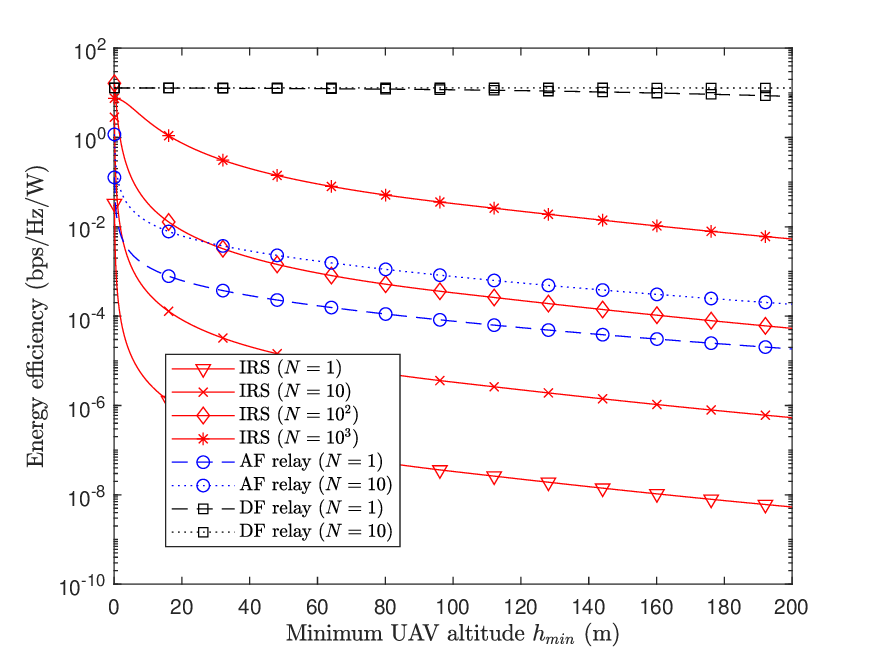}
		\label{fig:energy efficiency versus h_min}}
	\caption{Performance comparison between IRSs and full-duplex relays}
	\label{fig:influences of different levels of self-interference cancellation factors}
	\vspace{-1em}
\end{figure*}

Then, we compare IRSs against full-duplex relays under different sizes to draw some useful insights. It is illustrated in Fig. \ref{fig:minimum power versus distance} that higher transmit power is required to achieve a rate $R_0$ as $L$ increases. Also, an IRS of size $N=10^2$ is needed to outperform the AF relay of size $N=1$. On the other hand, IRSs are inferior to DF relays in our setup, even when the size are extremely large, i.e., $N=10^3$. It is shown in Fig. \ref{fig:energy efficiency versus achievable rate} that medium size ($N=10^2$) is needed by IRSs to outperform AF relays. Meanwhile, large size ($N=10^3$) as well as high data rates ($R_0\geq 16$bps/Hz) are required by IRSs to outperform DF relays of size $N = 1$, which coincides with the observation in\cite{bjornson2019intelligent}. Finally, the impact of $h_{min}$ is analyzed in Fig. \ref{fig:energy efficiency versus h_min}. We can see that the energy efficiency of all relaying schemes decreases with $h_{min}$. In particular, DF relay presents a rather small reduction in energy efficiency compared with IRS and AF relay.
\vspace{-0.2em}
\section{Conclusion}
\label{sec:conclusion}
In this article, we analytically compare IRSs against full-duplex relays, specifically in UAV communication scenarios. Both AF and DF relaying schemes are considered for comparison. Further, 3D UAV deployment as well as transmit power are optimized. Our numerical results demonstrate that DF relays generally outperform IRSs, except in cases where the surface size of IRSs is large and data rates are high. Moreover, AF relays exhibit comparable performance to IRSs of medium size.
\appendices
\section{Proof of Lemma 1}
\label{proof:lemma 1}
First, we readily have $x_U^*\in[0,L]$, $y_U^*=0$ and $h_U^*=h_{min}$ based on \cite[Lemma~3]{zeng2016throughput}. Then the problem is simplified as finding $x_U\in [0,L]$ to minimize $\tilde{f}(x_U)=(x_U^2 +h_{min}^2)((x_U-L)^2 +h_{min}^2)$. Following \cite[Theorem~1]{ren2023ondeployment}, we compute the first derivative of $\tilde{f}(x_U)$ as
\begin{equation}
	\begin{aligned}
		\frac{d\tilde{f}(x_U)}{dx_U}&=4x_U^3-6Lx_U^2+(2L^2 + 4h_{min}^2)x_U-2h_{min}^2L\\
	&=4\bigg(x_U-\frac{L}{2}\bigg)\bigg(\bigg(x_U-\frac{L}{2}\bigg)^2+h_{min}^2-\frac{L^2}{4}\bigg).
	\end{aligned}
\end{equation}

1) If $L\leq 2h_{min}$, we have $\frac{d\tilde{f}(x_U)}{dx_U}<0$ for $x_U \in[0,\frac{L}{2})$ and $\frac{d\tilde{f}(x_U)}{dx_U}>0$ for $x_U \in(\frac{L}{2},L]$. Obviously, $x_U =\frac{L}{2}$ is the minimum point of $\tilde{f}(x)$;

2) On the other hand, if $L> 2h_{min}$, we can obtain $x_{1,3}=\frac{L}{2} \pm\sqrt{\frac{L^2}{4}-h_{min}^2}$ and $x_2=\frac{L}{2}$ by setting $\frac{d\tilde{f}(x_U)}{dx_U}=0$. It is observed that $\frac{d\tilde{f}(x_U)}{dx_U}<0$ for $x_U \in[0,x_1)\cup(x_2,x_3)$ and $\frac{d\tilde{f}(x_U)}{dx_U}>0$ for $x_U \in(x_1,x_2)\cup (x_3,L]$. Therefore, $x_1$ and $x_3$ are two minimum points and $\tilde{f}(x_1)=\tilde{f}(x_2)$ since $\tilde{f}(x_U)$ is symmetric about $x_U=\frac{L}{2}$.
\vspace{-0.2em}
\section{Proof of Proposition 1}
\label{proof:proposition 1}
\eqref{eq:minimum transmit power for IRS} can be directly obtained by solving $R_{IRS}=R_0$. In the sequel, we focus on AF and DF relays. 

1) For the AF relay case, by setting $R_{AF}=R_0$, we shall obtain $p_S^{(AF)}=\frac{\rho(2^{R_0}-1)}{g_{SR}}p_R^{(AF)}+\frac{\sigma^2(2^{R_0}-1)}{N^2g_{SR}g_{RD}p_R^{(AF)}}+\frac{\sigma^2(2^{R_0}-1)}{g_{SR}}$. Therefore, $p_S^{(AF)}+ p_R^{(AF)}=\frac{\rho(2^{R_0}-1)+g_{SR}}{g_{SR}}p_R^{(AF)}+\frac{\sigma^2(2^{R_0}-1)}{N^2g_{SR}g_{RD}p_R^{(AF)}}+\frac{\sigma^2(2^{R_0}-1)}{g_{SR}}\overset{(b)}{\geq}2\sqrt{\frac{\sigma^2(\rho(2^{R_0}-1)+g_{SR})(2^{R_0}-1)}{N^2g_{SR}^2g_{RD}}} + \frac{\sigma^2(2^{R_0}-1)}{g_{SR}}$, where $(b)$ is obtained with Cauchy inequality and the equality holds when $p_R^{(AF)}=\sqrt{\frac{\sigma^2(2^{R_0}-1)}{N^2g_{RD}(\rho(2^{R_0}-1)+g_{SR})}}$.

2) For the DF relay case, \added{$\frac{p_S^{(DF)}g_{SR}}{\rho p_R^{(DF)}+\sigma^2}=\frac{Np_R^{(DF)}g_{RD}}{\sigma^2}$} must hold since otherwise we could further decrease $p_S^{(DF)}$ or $p_R^{(DF)}$ to achieve the equality without violating rate requirement $R_0$. Thereby, using $R_{DF}=R_0$ and after some algebra, we have $p_S^{(DF)}=\added{\frac{\sigma^2\rho(2^{R_0}-1)^2}{Ng_{SR}g_{RD}}+\frac{\sigma^2(2^{R_0}-1)}{g_{SR}}}$ and $p_R^{(DF)}=\added{\frac{\sigma^2(2^{R_0}-1)}{Ng_{RD}}}$.
\vspace{-0.2em}
\bibliographystyle{IEEEtran}
\bibliography{ref}

\vfill

\end{document}